\documentclass[onecolumn, draftcls]{ieeetran}
\usepackage{amsmath,amsthm,amssymb,paralist,subfigure,cite,graphicx,algorithm,url,tikz}
\newtheorem{lem}{Lemma}
\newtheorem{theorem}{Theorem}
\newtheorem{defn}{Definition}
\newtheorem{rem}{Remark}
\newtheorem{prop}{Proposition}
\newtheorem{cor}{Corollary}

\def\mb{\mathbf}
\def\mc{\mathcal}
\def\wh{\widehat}
\def\wt{\widetilde}

\title{\huge Graphic-theoretic distributed inference in social networks}
\author{Mohammadreza Doostmohammadian and Usman A. Khan
\thanks{
Department of Electrical and Computer Engineering, Tufts University, {\texttt{\{mrd,khan\}@ece.tufts.edu}}. This work has been partially supported by an NSF CAREER award \# CCF-1350264.}}

\begin{document}
\maketitle

\begin{abstract}
We consider distributed inference in social networks where a phenomenon of interest evolves over a given social interaction graph, referred to as the \emph{social digraph}. For inference, we assume that a network of agents monitors certain nodes in the social digraph and no agent may be able to perform inference within its neighborhood; the agents must rely on inter-agent communication. The key contributions of this paper include:
\begin{inparaenum}[(i)]
\item a novel construction of the distributed estimator and distributed observability from the first principles;
\item a graph-theoretic agent classification that establishes the importance and role of each agent towards inference;
\item characterizing the necessary conditions, based on the classification in (ii), on the agent network to achieve distributed observability.
\end{inparaenum}
Our results are based on structured systems theory and are applicable to any parameter choice of the underlying system matrix as long as the social digraph remains fixed. In other words, any social phenomena that evolves (linearly) over a structure-invariant social digraph may be considered--we refer to such systems as Liner Structure-Invariant (LSI). The aforementioned contributions, (i)--(iii), thus, only require the knowledge of the social digraph (topology) and are independent of the social phenomena. We show the applicability of the results to several real-wold social networks, i.e. social influence among monks, networks of political blogs and books, and a co-authorship graph.

\end{abstract}

\section{Introduction}
Social networks appear in a wide variety of contexts ranging from, e.g. humans, animals, communities, to economics, markets, sales, blogs, and citations; relevant literature includes~\cite{liu-pnas,Liu-nature,Birla:wo,acc13_ghaderi,Friedkin:1999vd,acemoglu2011opinion,jadbabaie2012belief,friedkin2006structural,newman2006barabasi,Xing:2010ft,Friedkin:2003tx,opinion2002hegselmann} and references therein. The underlying social phenomena of interest (evolving on such networks) also vary widely and include voting models, flocking, herd behavior, rumor propagation, stock prices, and (community) trends to name a few. An additional layer of complexity is added by noting that the associated dynamics can be either non-linear or linear. The problem of inference in social networks, thus, is highly complex as it has to tackle the wide diversity and range of the underlying social dynamics. What is uniform, however, across all of the social networks and the corresponding dynamics is the presence of a \emph{social digraph}, i.e. the social interactions over which a phenomenon of interest evolves. In this context, we formulate the inference problem, and subsequently the controllability and observability of social dynamics, on the social digraphs, independent of the particular social phenomenon. Our prime focus is on linear dynamics, however, as discussed in~\cite{liu-pnas,Liu-nature}, observability (and controllability) of nonlinear dynamics can also be characterized explicitly on the underlying system digraphs, e.g. the observability of bio- and chemical-networks based on the state-interactions. 
\begin{figure}
\centering
\includegraphics [height=1.45in]{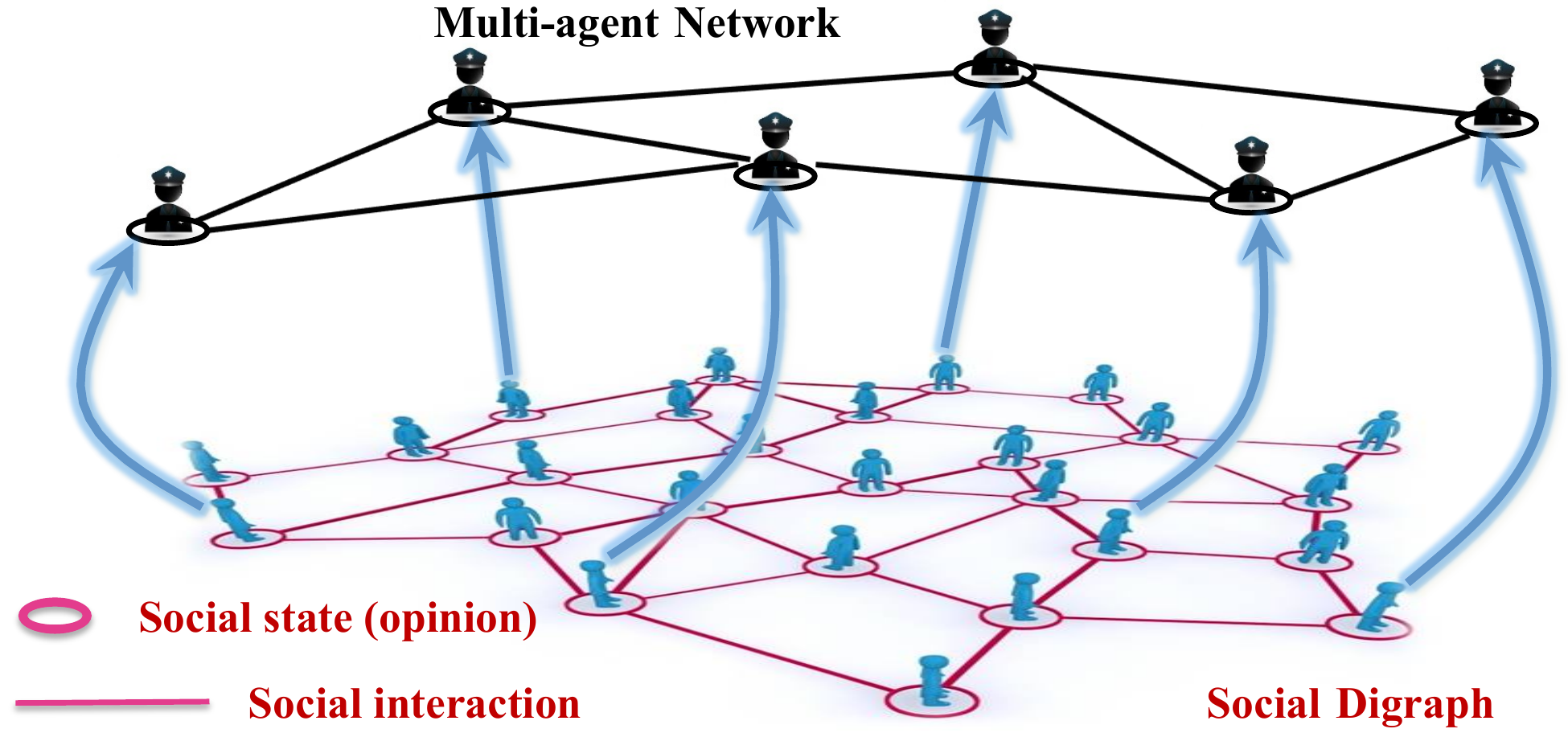}
\caption{Social network of \textit{dynamic} influence and fixed interactions with agents monitoring the nodes (states) of the interaction graph}
\label{figsocial}
\end{figure}

This paper studies distributed inference of social dynamics that naturally evolve over the social digraphs. Each state within the social dynamics is a node in the corresponding digraph and may represent an opinion~\cite{acemoglu2011opinion}, or a belief~\cite{jadbabaie2012belief}, of an individual, actor, or a member of the social network. The states evolve (linearly or non-linearly) over the social digraph, i.e. according to the interactions among the members of the social network, see e.g.~\cite{friedkin2006structural,newman2006barabasi} for details on relevant social models. In this sense, \emph{a social state} and \emph{a node in the social digraph} will be used interchangeably. For the purpose of inference, we deploy a network of agents that monitor a few nodes (social states) of the social digraph, see Fig.~\ref{figsocial}. Since the agent measurements may not be sufficient to build the entire state estimate, the agents communicate among themselves. A natural question is to \emph{design communication among the agents} that results in the \emph{distributed observability} of the social digraph; in a way that the results are independent of the particular social dynamics and only depend on the social interactions (digraph). 

We study communication design towards distributed observability over a Linear Structure-Invariant (LSI) characterization of social networks. An LSI system is such that the structure (zero/non-zero pattern) of its system matrices remains fixed but the non-zero elements may take any arbitrary (non-zero) values possibly changing over time. In this sense, \emph{our results are applicable to any social phenomena described over a given social digraph}. This is because any member of the social network can change the weight associated to its neighboring actors (in the social digraph) but some weight will always be assigned, e.g. see~\cite{Xing:2010ft,Friedkin:2003tx}. Related literature on linear social models includes: Reference~\cite{acemoglu2011opinion} on opinion dynamics; Reference~\cite{Friedkin:1999vd} on actor influence networks along with~\cite{acc13_ghaderi} on actor stubbornness; Reference~\cite{acemoglu2011opinion,opinion2002hegselmann} on Markov opinions; and, References~\cite{acemoglu2011opinion,Birla:wo,acc13_ghaderi,Friedkin:1999vd} on consensus-based models. Towards distributed estimation, related work includes the earlier work on Kalman-consensus filters~\cite{olfati:05,khan-DKF,Giannakis-Estimation} to more recent work on moving horizon estimators~\cite{farina-mhe} and distributed Kalman filters~\cite{usman_cdc:10,jstsp}. Another important object of study has been the characterization of distributed observability~\cite{nuno-suff.ness,battistelli_cdc,ugrinovskii2012conditions}. 

The main contributions of this paper includes the following towards social phenomena modeled by discrete-time LSI systems (perhaps, after appropriate linearization~\cite{liu-pnas,Liu-nature}): \emph{First}, we mathematically formulate distributed observability involving the Kronecker product of the agent (communication) topology and the social digraph (the LSI system matrix); and, \emph{Second}, we derive the \textit{necessary} conditions on the agent topology to ensure distributed observability. During this process, we specifically address the following:
\begin{inparaenum}[(i)]
\item Which states are critical for the inference of social dynamics? and,
\item Given the critical states observed by agents, what are the agent connectivity requirements to ensure distributed observability?
\end{inparaenum}
The first question aims at studying the contribution of each observation towards the centralized observability. As we will show, the observations critical for the centralized observability are also critical for distributed observability. However, the communication among the agents (possessing these observations) is different. In particular, we show that each critical observation is not required given certain agent connectivity.

We treat distributed observability as \emph{generic},~i.e. only tied to the structure of the LSI system matrix, and, in turn, to the social digraph. Our work significantly differs from the related work on centralized observability~\cite{commault-recovery,boukhobza-recovery,liu-pnas,acc13_kar}, and its dual on centralized controllability~\cite{Liu-nature,woude:03}, due to the following reasons: 
\begin{inparaenum}[(a)]
\item We show that the critical measurements, e.g. in~\cite{commault-recovery}, can be further partitioned into two classes: Type-$\alpha$ and Type-$\beta$, driven by the structural rank of the LSI system matrix (social digraph); 
\item This partitioning further enables us to show: 
\begin{inparaenum}[(i)]
\item Only Type-$\alpha$ measurements are required at each agent; and,
\item Type-$\beta$ measurements are not necessarily required as long as each agent has a directed path to all Type-$\beta$ measurements;
\end{inparaenum}
\item Our analysis does not make any assumption on the structural-rank of the social system in contrast to~\cite{icassp13,roy_defective,battistelli_cdc}, where full structural-rank is assumed;
\item The results are also distinguishable from~\cite{hierarchy-giannakis,hierarchy-egerstedt}, where hierarchical agent topology is assumed. 
\end{inparaenum}
The main objective here is to \textit{design} the agent topology in contrast to predetermined networks in~\cite{usman_cdc:11,farina-mhe,battistelli_cdc,sayedtu12}.

The rest of the paper is organized as follows: Section~\ref{pr} gives preliminaries on graph theory, social modeling, and formulates the distributed estimation problem. Section~\ref{dist-prob} provides a mathematical derivation of distributed observability. Section~\ref{obsrec} enlists some advanced graph-theoretic concepts subsequently used in Sections~\ref{nessmeas} and~\ref{partial} to derive necessary observations and agent connectivity. We illustrate our approach on examples of real-world social networks in Section~\ref{example}. Finally, Section~\ref{con} concludes the paper.

\section{Preliminaries and Problem formulation} \label{pr}
In this section, we discuss the preliminaries to describe this paper and formulate the distributed inference problem.

\subsection{Social phenomenon: Modeling and interactions} \label{mdl}
Social networks and complex networks, in general, have been modeled using both linear and nonlinear dynamics, see~\cite{blondel2009krause ,friedkin2006structural, newman2006barabasi}, and references within. Examples of linear models are in consensus/agreement problems~\cite{Birla:wo, acc13_ghaderi,Friedkin:1999vd} and Markov-based opinion formation~\cite{acemoglu2011opinion,opinion2002hegselmann}. Two well-known linear models are social influence networks by Freidkin and Johnson~\cite{Friedkin:1999vd} and French model~\cite{French}. The French model formulates the formation of opinions (states) under the interpersonal influence of peers. Similarly, Freidkin and Johnson model the process of social influence on opinion evolution. Another socio-economic example is~\cite{vanDalen:2010uq}, where product prices as states linearly evolve on a daily basis according to a competitiveness matrix (auction game). Of significant relevance to this paper is the characterization in~\cite{Xing:2010ft} and~\cite{Friedkin:2003tx} where the structure of the linear model is assumed to be fixed but with time-varying interaction weights. In particular, Reference~\cite{Xing:2010ft} describes examples of a linear state-space on the social networks resulting from email communication, and social interaction of Monks (members of a particular religious order). On the other hand, Reference~\cite{Friedkin:2003tx} discusses a linear state-space for influence networks, where attitudes, sentiments, or expectations (states) evolve over time-varying influences of other actors.

For nonlinear social dynamics, simplified modeling methods have been considered, e.g.~\cite{liu-pnas,Liu-nature}. Particularly, observability of nonlinear dynamics is characterized by the structural observability of the corresponding linearized system~\cite{liu-pnas}. Hence, it is natural to model the social phenomena as LSI systems, where any (non-zero) element of the system matrix may change (modeling distinct or time-varying phenomena) as long as the structure (social digraph) is not violated, e.g. see~\cite{Xing:2010ft} and~\cite{Friedkin:2003tx}. Mathematically, we model the social dynamics as 
\begin{eqnarray}\label{sys1} \label{A}
\mb{x}_{k+1} = A_k\mb{x}_k + \mb{v}_k,\qquad k\geq0,
\end{eqnarray}
where~$A_k$ is the system matrix,~$\mb{x}_k=[x_{1,k}~\ldots~x_{n,k}]^T\in\mathbb{R}^n$ is the state vector, and~$\mb{v}_k$ is Gaussian noise. \emph{The system matrix,~$A_k$, is such that its elements can change but its structure, denoted by~$A=\{a_{ij}\}$ and following from the social digraph, is invariant over time.} As an example, consider the email communication network in~\cite{Xing:2010ft} where the states are time-series of the email communication among ENRON employees. The number of emails exchanged are modeled as a linear state-space and their evolution before and after collapse of the company is studied. Another related LSI description is the monk network where a group of monks are ranked based on their inter-relations (e.g. liking/disliking, praise/blame, etc.) form the well-known Sampson's network~\cite{sampson_network}.

We assume that Eq.~\eqref{sys1} is monitored by~$N$ agents:
\begin{eqnarray} \label{H_i}
\mb{y}_k^i = H^i\mb{x}_k + \mb{r}_k^i.
\end{eqnarray}
where~$H^i_k=\{h_{ij}\}\in\mathbb{R}^{p_i\times n}$ is the local observation matrix at agent~$i$ and time~$k$;~$\mb{y}^i_k\in\mathbb{R}^{p_i}$  is the local observation vector, and~$\mb{r}^i_k$ is local observation Gaussian noise. With this notation, the global observation model is
\begin{eqnarray}\nonumber
\left[
\begin{array}{c}
\mb{y}^1_{k}\\
\vdots\\
\mb{y}^N_{k}
\end{array}
\right] &=&
\left[
\begin{array}{c}
H_{1}\\
\vdots\\
H_{N}
\end{array}
\right]\left[
\begin{array}{c}
\mb{x}^1_{k}\\
\vdots\\
\mb{x}^N_{k}
\end{array}
\right]+
\left[
\begin{array}{c}
\mb{r}^1_{k}\\
\vdots\\
\mb{r}^N_{k}
\end{array}
\right],\\
\triangleq \mb{y}_k &=& H\mb{x}_k + \mb{r}_k,
\end{eqnarray}
where~$\mb{y}_k\in\mathbb{R}^{p}$ is the global observation vector,~$H=\{h_{ij}\}\in\mathbb{R}^{p\times n}$ is the global observation matrix, and~$\mb{r}_k$ is noise. Clearly, we have~$p=p_1+\ldots+p_N$.

The agents,~$i=1,\ldots,N$, monitoring the social network, exchange information over a communication graph,~$\mc{G}_W=(\mc{V}_W,\mc{E}_W)$. The set,~$\mc{V}_W$, consists of all of the agents, whereas the set,~$\mc{E}_W$, is the edge set of ordered pairs,~$(i,j)$, describing that agent~$j$ can send information to agent~$i$. The neighborhood at an agent~$i\in\mc{V}_W$ is denoted by~$\mc{N}_i$, defined as
\begin{eqnarray}
\mc{N}_i=\{j~|~(i,j)\in\mc{E}_W\},
\end{eqnarray}
where we use~$\mc{D}_i$ to denote the extended neighborhood, i.e.~$\{i\}\cup\mc{N}_i$, and~$\mbox{Adj}(\mc{G}_W)$ denotes the graph adjacency.

\emph{We are explicitly interested in designing the structure of $H$, i.e. where to place the agents, and the structure of $\mbox{Adj}(\mc{G}_W)$, i.e. how the agent should communicate, given only the structure of the LSI system matrix, $A$, i.e. the social digraph.}

\subsection{Structured Systems Theory}\label{sst_pre}
{\bf System (social) digraph}: In structured systems theory, the system in Eqs.~\eqref{sys1}-\eqref{H_i} is typically modeled as a \emph{system digraph}, where the nodes are states and the edges are the social interactions given by the system matrix,~$A=\{a_{ij}\}$. Let~$\mc{X}\triangleq\{x_1,\ldots,x_n\}$ denote the set of states, and~$\mc{Y}\triangleq\{\mb{y}^1,\ldots,\mb{y}^N\}$ denote the set of observations. Then the system digraph is given by~$\mc{G}_A=(\mc{V}_A=X,\mc{E}_A)$. The edge set,~$\mc{E}_A$, is defined as~$\mc{E}_A=\{(x_i,x_j)~|~a_{ij}\neq0\}$ to be interpreted as~$x_i\leftarrow x_j$. In our social model,~$\mc{G}_A$ precisely captures the interactions on the social digraph over which the social phenomena evolve.

{\bf Composite digraph}: A digraph,~$\mc{G}_{\scriptsize \mbox{sys}} = (\mc{V}_{\scriptsize \mbox{sys}},\mc{E}_{\scriptsize \mbox{sys}})$, described on both states and observations:~$\mc{V}_{\scriptsize \mbox{sys}}=\{\mc{X} \cup \mc{Y}\}$, and~$\mc{E}_{\scriptsize \mbox{sys}}=\{(x_i,x_j)~|~a_{ij}\neq0\} \cup \{(y_i,x_j)~|~h_{ij}\neq0\}$. This composite graph,~$\mc{G}_{\scriptsize \mbox{sys}}$, is associated to the pair~$(A,H)$ and adds observations (agents) to the system (social) digraph,~$\mc{G}_A$. The following are defined over the composite graph,~$\mc{G}_{\scriptsize \mbox{sys}}$.

A \emph{path},~$i\overset{\scriptsize\mbox{path}}{\longrightarrow} j$, from~$i\in\mc{V}_{\scriptsize \mbox{sys}}$ to~$j\in\mc{V}_{\scriptsize \mbox{sys}}$, is such that there exists a sequence of nodes,~$\{i,{i_1},\ldots,i_{L-1},j\}$ in~$\mc{V}_{\scriptsize \mbox{sys}}$, with~$(j,i_{L-1}),\ldots,(i_1,i)\in\mc{E}_{\scriptsize \mbox{sys}}$. A path is called \emph{$\mc{Y}$-connected}, if it ends at a node in~$\mc{Y}$. A \emph{cycle} is a path where the begin and end nodes are the same. With this notation, the following theorem~\cite{liu-pnas} states the conditions required for \textit{structural centralized observability} of the social digraph:
\begin{theorem}\label{th1}
A system is observable if and only if in,~$\mc{G}_{\scriptsize \mbox{sys}}$: \textbf{Accessibility}--Each state is the begin-node of a~$\mathcal{Y}$-connected path; and, \textbf{$S$-rank condition}--There exist a union of \textit{disjoint}~cycles and~$\mathcal{Y}$-connected paths covering all the states.
\end{theorem}

It is noteworthy that Theorem~\ref{th1} does not require the exact system parameters but only the zero/non-zero structure. In this sense, the observability characterization of social systems is \emph{generic}, i.e. Theorem~\ref{th1} is applicable to any choice of social dynamics with a given social digraph, structure of~$A$ and~$H$--the set of values for which Theorem~\ref{th1} does not hold lie on an algebraic variety whose Lebesgue measure is zero~\cite{woude:03}.

The accessibility and~$S$-rank conditions of Theorem~\ref{th1} also have algebraic interpretations~\cite{shields}. \emph{Accessibility} is tied with the irreducibility of the structure of~$\left[A^\top~H^\top\right]^\top$. Intuitively, each state (node in social digraph) in~$\mc{G}_A$ must have its own `downstream' observation. If a state has no downstream state or observation, the information of that state (and its `upstream' nodes) are not accessible. The~$S$\emph{-rank condition} is equivalent to~$\mbox{$S$-rank}\left[A^\top~H^\top\right]^\top=n$. In other words, if two or more states point to the same downstream node, the distinct information of all the upstream nodes cannot be inferred from the same downstream node. 

Graph-theoretic notions of contractions~\cite{commault-recovery,boukhobza-recovery}, and Strongly Connected Components (SCC)~\cite{asilomar11} are useful for such structured analysis. Some related concepts will be introduced in Sections~\ref{agr1} and~\ref{agr2}, whereas, Section~\ref{DMalg} provides a list of relevant computationally-efficient tools.

\subsection{Problem formulation: Distributed inference}
The problem of estimating the state,~$\mb{x}_k$ from the distributed agent observations,~$\mb{y}_k^i$, can be fundamentally considered in two different contexts:
\begin{inparaenum}[(i)]
\item Central--the agent observations are collected at a center where the estimate is computed~\cite{bay}; and,
\item Distributed--the agents interact with each other over the communication graph,~$\mc{G}_W$, and each agent estimates the entire state,~$\mb{x}_k$, given its observations \emph{and} the interactions up to time~$k$. This estimate is denoted by~$\wh{\mb{x}}_{k|k}^i$ at agent~$i$.
\end{inparaenum}

In this paper, we explicitly consider the distributed estimation assuming only that the pair,~$(A,H)$, is observable~\cite{bay}, i.e. no agent may estimate the entire state within its neighborhood. In other words, neither the pair,~$(A,H_i)$, nor the pair,~$(A,\{H_j\}_{j\in\mc{D}_i})$, may be observable. It is noteworthy that unlike many estimation schemes~\cite{olfati:05,khan-DKF, Giannakis-Estimation}, we assume that each agent exchanges information over~$\mc{G}_W$ only once per~$k$. Under these assumptions, we consider the following problems in this paper. Given the social digraph, $\mc{G}_A$, Eqs.~\eqref{sys1}-\eqref{H_i}, and the agent communication,~$\mc{G}_W$:
\begin{enumerate}[(a)]
\item What is distributed observability of a social digraph? We derive this in Section~\ref{dist-prob} leading to a distributed estimator.
\item What are the necessary conditions on the agent communication,~$\mc{G}_W$, such that the underlying social digraph is distributedly observable?
\end{enumerate}

\section{Distributed observability} \label{dist-prob}
We now describe distributed estimation in more detail where each agent~$i$ is to estimate the state,~$\mb{x}_k$, with its observations,~$\mb{y}_k^i$, and with its neighboring observations,~$\{\mb{y}_k^j\}_{j\in\mc{N}_i}$. Each agent,~$i$, thus, estimates the state-vector, described by Eq.~\eqref{sys1}, from the following observations:
\begin{eqnarray} \label{sys3}
\mb{y}_k^j = H_j\mb{x}_k + \mb{r}_k^j,\qquad j\in\{i\}\cup\mc{N}_i.
\end{eqnarray}
Let us assume that the neighbor set has a total of~$N_i$ neighbors, i.e.~$|\mc{N}_i|=N_i$, and is indexed by~$i_1,i_2,\ldots,i_{N_i}$. Then, agent~$i$ is to estimate~$\mb{x}_k$ from the neighboring observations,~$\mb{y}_k^i,\mb{y}_k^{i_1},\ldots,\mb{y}_k^{i_{N_i}}.$ Or, equivalently, with the following:
\begin{eqnarray}\label{Usys1}
\wt{\mb{y}}_k^i &\triangleq& \left[
\begin{array}{c}
\mb{y}_k^i\\
\vdots\\
\mb{y}_k^{i_{N_i}}
\end{array}
\right]=
\left[
\begin{array}{c}
H_i\\
\vdots\\
H_{N_i}
\end{array}
\right]\mb{x}_k +
\left[
\begin{array}{c}
\mb{r}_k^i\\
\vdots\\
\mb{r}_k^{N_i}
\end{array}
\right].
\end{eqnarray}
The above observation model is equivalent to~\cite{usman_acc:11}:
\begin{eqnarray}\label{Usys2}
\mb{z}_k^i=\left[
\begin{array}{cccc}
H_i^\top&\ldots&H_{N_i}^\top
\end{array}
\right]\wt{\mb{y}}_k^i\triangleq \wt{H}_i\mb{x}_k + \wt{\mb{r}}_k^j,\\
\mbox{with}~~~\wt{H}_i\triangleq \sum_{j\in\{i\}\cup\mc{N}_i}H_j^\top H_j,\qquad\wt{\mb{r}}_k^i\triangleq\sum_{j\in\{i\}\cup\mc{N}_i}H_j^\top\mb{r}_k^j.
\end{eqnarray}
In fact, Eq.~\eqref{Usys2} is just a compact way of writing Eq.~\eqref{Usys1}. The distributed estimation problem over the communication graph,~$\mc{G}_W$, is now to estimate~$\mb{x}_k$ at each agent,~$i$, with the observations,~$\mb{z}_k^i$. From the standard estimation theory arguments~\cite{bay}, we know that such an estimation is possible at any agent~$i$, if and only if, the pair,~$(A,\wt{H}_i)$, is observable. For observability at all of the agents, we must consider all such pairs,~$(A,\wt{H}_1), (A,\wt{H}_2), \ldots, (A,\wt{H}_N)$, i.e. the observability of
\begin{eqnarray}
\left(
\left[
\begin{array}{cccc}
A\\
&\ddots\\
&&A
\end{array}
\right],\underbrace{\left[
\begin{array}{cccc}
\wt{H}_1\\
&\ddots\\
&&\wt{H}_N
\end{array}
\right]}_{\triangleq D_H}\right),
\end{eqnarray}
compactly written as~$(I\otimes A, D_H)$. It is straightforward to show that a centrally observable system does not necessarily imply that the distributed system is also observable, i.e.
\begin{eqnarray}
(A,H)\mbox{-observability}\nRightarrow(I\otimes A,D_H)\mbox{-observability}.
\end{eqnarray}

We note that the above straightforward description of distributed observability is actually misleading. The primary reason is that although observation exchanges are considered, the agents may also exchange their local predictors. This latter exchange does not appear in the above characterization of distributed observability. In the following, we provide a novel construction to derive distributed observability that accommodates for both observation and predictor exchanges, and show that distributed observability \emph{does not} require each agent to be observable in its neighborhood.

\subsection{Derivation}\label{do_d}
Consider again the distributed estimation problem where we wish to estimate the dynamics in Eq.~\eqref{sys1} via the observations in Eq.~\eqref{sys3}. Recall that~$\widehat{\mb{x}}^i_{k|k}$ denotes the estimate of the state,~$\mb{x}_k$, using all of the observations available at agent~$i$, and its neighboring agents up to time~$k$. Concatenating the estimates at all agents, the \emph{global} state estimate in the network is
\begin{eqnarray}
\underline{\widehat{\mb{x}}}_{k|k}\triangleq
\left[
\begin{array}{c}
\widehat{\mb{x}}^1_{k|k}\\
\widehat{\mb{x}}^2_{k|k}\\
\vdots \\
\widehat{\mb{x}}^N_{k|k}
\end{array}
\right].
\end{eqnarray}
Considering~$\underline{\widehat{\mb{x}}}_{k|k}$ to be an estimate of some state, we seek the corresponding dynamical system to this state-estimate. Clearly, the corresponding dynamical system has the following \emph{global} state vector:
\begin{eqnarray}\label{gst}
\underline{\mb{x}}_{k} \triangleq \left[
\begin{array}{cccc}
\mb{x}_k^\top & \mb{x}_k^\top & \ldots &\mb{x}_k^\top
\end{array}
\right]^\top=\mb{1}_N\otimes \mb{x}_k,
\end{eqnarray}
where~$\mb{1}_N$ is a column vector of~$N$ ones. To this end, let us assume that the dynamics associated to the above \emph{global} state-vector,~$\underline{\mb{x}}_{k}$, are given by some linear system:
\begin{eqnarray}\label{sys5}
\underline{\mb{x}}_{k+1} = Z \underline{\mb{x}}_{k} + \underline{\mb{v}}_{k},
\end{eqnarray}
where we have~$Z\in\mc{Z}$, and~$\mc{Z}$ is defined as a \emph{class} of system matrices such that if we choose any matrix~$Z\in\mc{Z}$, Eq.~\eqref{sys5} remains a valid representation of the global state vector as given by concatenating the system dynamics of Eq.~\eqref{sys1}. We now characterize this class of system matrices,~$\mc{Z}$. We have
\begin{equation}
\underline{\mb{x}}_{k+1} = \mb{1}_N\otimes (A\mb{x}_{k}+\mb{v}_k)
= \underbrace{(W\otimes A)}_{\mc{Z}}\mb{x}_{k}+\underbrace{\mb{1}_N\otimes\mb{v}_k}_{\underline{\mb{v}}_{k}},
\end{equation}
where the last equality follow if and only if~$W$ is stochastic. It is, in fact, quite straightforward to show that
\begin{eqnarray}
\mb{1}_N\otimes A\mb{x}_k = (W\otimes A)\mb{x}_k,
\end{eqnarray}
for any stochastic matrix,~$W$, leading to the conclusion that any matrix that cannot be decomposed as~$W\otimes A$ is not a system matrix for the dynamics described by~$\underline{\mb{x}}_{k+1}$, i.e.
\begin{eqnarray}
\mc{Z} = \{Z~|~Z=(W\otimes A) ~ \mbox{and}~ W \mbox{ is stochastic}\}.
\end{eqnarray}
The propositions below follow from the above arguments.
\begin{prop}\label{p0}
The distributed estimation of the dynamics in Eq.~\eqref{sys1} monitored by agents according to Eq.~\eqref{H_i}, interacting over a communication graph,~$G_W$, is equivalent to the centralized estimation of the following system:
\begin{eqnarray}\label{cf1}
\underline{\mb{x}}_{k+1} &=& (W\otimes A)\underline{\mb{x}}_{k} + \underline{\mb{v}}_{k+1},\\\label{cf2}
\mb{z}_k &\triangleq & D_H\underline{\mb{x}}_{k} + \widetilde{\mb{r}}_k,
\end{eqnarray}
where~$W$ is stochastic and is such that its sparsity (zero/non-zero pattern) is the same as of the adjacency matrix,~$\mbox{Adj}(\mc{G}_W)$.
\end{prop}
We are now in a position to write the optimal filtering equations for the centralized system (equivalent to the distributed estimation problem) in Eqs.~\eqref{cf1}-\eqref{cf2}:
\begin{eqnarray}\label{cfe1}
\wh{\underline{\mb{x}}}_{k|k-1} &=& (W\otimes A)\wh{\underline{\mb{x}}}_{k-1|k-1},\\\label{cfe2}
\wh{\underline{\mb{x}}}_{k|k} &=& \wh{\underline{\mb{x}}}_{k|k-1} + \underline{K}_k\left(\mb{z}_k - D_H\wh{\underline{\mb{x}}}_{k|k-1}\right),
\end{eqnarray}
where~$\underline{K}_k$ is the Kalman gain. The following proposition formally defines the distributed observability.
\begin{prop}\label{1}
A dynamical system monitored by a network of interacting agents is \textit{distributively observable} if and only if~$(W \otimes A, D_H)$ is observable, where~$W$ is a stochastic matrix,~$W$, and has the same sparsity as that of the adjacency matrix,~$\mbox{Adj}(\mc{G}_W)$.
\end{prop}
\begin{proof}
The proof relies on the fact that the distributed estimation problem is equivalent to the centralized estimation problem with the system matrices,~$W\otimes A$ and~$D_H$.
\end{proof}
In general the observability of the pair~$(W \otimes A, D_H)$ can be checked using the algebraic observability tests, i.e. the rank of the observability Grammian or the PBH test~\cite{bay}. These tests, however, require the explicit knowledge of the social phenomenon, i.e. the elements in the associated matrices. As we are concerned with LSI systems and social networks with LSI description, we are interested in developing observability tests that are based on the structure of the social digraph. The structure of the matrix~$W$ that makes~$(W \otimes A, D_H)$ observable, thus defines the topology of the underlying agent communication,~$\mathcal{G}_W$, see Proposition~\ref{p0}. In this paper, we will derive the necessary conditions on the communication topology,~$\mc{G}_W$, to recover this distributed observability. 

\subsection{Distributed estimator} \label{estimator}
Although the centralized system, Eqs.~\eqref{cf1}-\eqref{cf2}, is equivalent to the distributed estimation problem, we still have to verify that the centralized (optimal) filtering equations, Eqs.~\eqref{cfe1}-\eqref{cfe2}, can be implemented in a distributed fashion. To this end, we note that Eqs.~\eqref{cfe1}-\eqref{cfe2} consists of two information fusions: one is fusion in the predictor space, i.e. via~$W$ in Eq.~\eqref{cfe1}, and the other is the fusion in the observation space, i.e. via~$D_H$ in Eq.~\eqref{cfe2}. When the goal is to design a communication graph, it is advantageous to consider these two fusions separately. 

Following the above arguments, we consider the fusion in the predictor space to be implemented over an inter-agent communication graph,~$\mc{G}_\beta$, and the fusion in the observation space to be implemented over an inter-agent communication graph\footnote{Considering fusion over separate graphs is important because different connectivity conditions may be required for each fusion. For example, we showed in~\cite{asilomar11} that when the system matrix,~$A$, is full-rank, observation fusion is not required.},~$\mc{G}_\alpha$. We call~$\mc{G}_\alpha$ and~$\mc{G}_\beta$ respectively~$\alpha$-network and~$\beta$-network. With this two-layered approach to fusion, we can immediately note that~$\mc{G}_W$ is now given by~$\mc{G}_\alpha\cup\mc{G}_\beta$. Finally, we denote the neighborhood at agent~$i$ as~$\mc{N}_\alpha(i)$ and~$\mc{N}_\beta(i)$, in~$\mc{G}_\alpha$ and~$\mc{G}_\beta$, respectively.

We first consider the prediction in Eq.~\eqref{cfe1}. Assume~$W$ to be a stochastic matrix such that the zero and non-zero pattern follow the sparsity of the adjacency matrix,~$\mbox{Adj}(\mc{G}_\beta)$.
It can be immediately observed that Eq.~\eqref{cfe1} is distributed:
\begin{eqnarray}\label{lp}
\widehat{\mb{x}}^i_{k|k-1} = \sum_{j\in\{i\}\cup\mathcal{N}_\beta(i)} w_{ij}A\widehat{\mb{x}}^j_{k-1|k-1},
\end{eqnarray}
with~$W=\{w_{ij}\}$. Next consider fusion in the observation space, i.e. Eq.~\eqref{cfe2}. Note that since the Kalman gain,~$\underline{K}_k$, is a full matrix in general, Eq.~\eqref{cfe2} cannot be immediately distributed. In order to keep the implementation of Eq.~\eqref{cfe2} distributed and local, an alternate is to assume that the gain matrix,~$\underline{K}_k$, is block-diagonal, i.e.~$\underline{K}_k=\mbox{blockdiag}[K_k^i,\ldots,K_k^N]$, leading to
\begin{eqnarray}\label{le}
\widehat{\mb{x}}^i_{k|k} =\widehat{\mb{x}}^i_{k|k-1} + K_k^i \sum_{j\in\{i\}\cup \mc{N}_\alpha(i)}H_j^\top \left(\mb{y}^j_k-H_j\widehat{\mb{x}}^i_{k|k-1}\right).
\end{eqnarray}

By restricting the full gain matrix,~$\underline{K}_k$, to be block-diagonal (or to any non-full structure), the resulting distributed estimator, Eqs.~\eqref{lp}-\eqref{le}, is not equal to the centralized counterpart, Eqs.~\eqref{cfe1}-\eqref{cfe2}. In other words, the Kalman gain matrix, $\underline{K}_k$, cannot be computed locally from the standard procedures. However, computing such a constrained gain is possible via an iterative cone-complementarity optimization algorithm, see~\cite{5717159, rami:97} for details. Nevertheless, if the centralized equivalent has no solution, then the distributed problem cannot have a solution and it is imperative to ensure the observability of~$(W \otimes A, D_H)$. 

Finally, it can be shown that the networked error in the distributed estimator, Eqs.~\eqref{lp}-\eqref{le}, evolves as
\begin{eqnarray}\label{err1}
\mb{e}_{k} = (W\otimes A - K_kD_H(W\otimes A))\mb{e}_{k-1} +
\mb{q}_k,
\end{eqnarray}
which is stable if and only if~$(W \otimes A, D_H)$ is observable~\cite{jstsp}. This is consistent with Proposition~\ref{1}.

\section{Recovering Observability:\\ A Graph Theoretic Approach} \label{obsrec}
We now focus on developing the necessary conditions on the agent communication graph,~$\mc{G}_W=\mc{G}_\beta\cup\mc{G}_\alpha$, in order to recover the distributed observability of the pair~$(W\otimes A, D_H)$. We cast this problem from a structural viewpoint (as introduced in Theorem~\ref{th1}), i.e. the analysis is irrespective of the particular social phenomena, elements in~$W\otimes A$ and~$D_H$, and only relies on the (composite) social digraph, i.e. the structure of the system matrix,~$W\otimes A$, and the observation matrix,~$D_H$. In order to develop our results in the structural context, we need some advanced graph theoretic concepts that are covered below in Sections~\ref{agr1} and~\ref{agr2}. These concepts provide the foundations and related preliminaries for the agent classification in Section~\ref{nessmeas}, and for the necessary conditions on designing the agent communication network in Section~\ref{partial}.

\subsection{Contractions in Bipartite graphs}\label{agr1}
The graph-theoretic concepts and notations stated in this section can be found in~\cite{murota}. We cast these definitions in our framework of system digraphs and illustrate them in Figs.~\ref{fig3node} (b)--(e) using a~$3$-node system digraph, $\mc{G}_A$ of Fig.~\ref{fig3node} (a).

{\bf Bipartite graphs}: A bipartite graph,~$\Gamma=(\mc{V}^+,\mc{V}^-,\mc{E}_\Gamma)$, is such that its nodes can be partitioned into two disjoint sets:~$\mc{V}^+$ and~$\mc{V}^-$, such that all of its edges~$\in\mc{E}_\Gamma$ start in~$\mc{V}^+$ and end in~$\mc{V}^-$. We construct a bipartite graph,~$\Gamma_A$ in Fig.~\ref{fig3node} (b), from the social digraph,~$\mc{G}_A$, of Fig.~\ref{fig3node} (a): Define~$\mc{V}^+= \mc{X}$ and~$\mc{V}^- = \mc{X}$, with the edge set,~$\mc{E}_{\Gamma_A},$ defined as the collection of~$(v_j^-,v_i^+)$, if ~$(v_j,v_i) \in \mc{E}_A$.

{\bf Matching}: A matching,~$\underline{\mc{M}}$, on the system digraph,~$\mc{G}_{A}$, is defined as a subset of the edge set,~$\mc{E}_A$, with no common end-nodes. In the bipartite graph,~$\Gamma_A$, it is defined as a subset of edges where no two of them are incident on the same node, i.e. all the edges in~$\mc{M}$ are all disjoint. The number of edges,~$|\underline{\mc{M}}|$, in~$\underline{\mc{M}}$ is the size of the matching. A matching,~$\underline{\mc{M}}$, with maximum size, is called maximal matching, denoted by~$\mc{M}$, which, is non-unique, in general. A maximal matching is shown as blue highlighted edges in Fig.~\ref{fig3node} (a) and (b).

{\bf Matched/Unmatched nodes}: Let~$\mc{M}$ be a maximal matching on~$\Gamma_{A}$. Let~$\partial \mc{M}^+$ and~$\partial \mc{M}^-$ denote the nodes incident to~$\mc{M}$ in~$\mc{V}^+$ and~$\mc{V}^-$, respectively. Then~$\delta \mc{M} = \mc{V}^+ \backslash \partial \mc{M}^+$ is the set of unmatched nodes, shown in a box in Fig.~\ref{fig3node} (b). 

{\bf Auxiliary graph}, denoted by~$\Gamma^\mc{M} _A$, is a graph associated to a maximal matching,~$\mc{M}$. It is constructed by reversing all the edges of maximal matching,~$\mc{M}$, and keeping the direction of all other edges, i.e.~$\mc{E}_{\Gamma_A} \backslash \mc{M}$, in~$\Gamma_A$, see Fig.~\ref{fig3node} (c).

{\bf Alternating path}: In the auxiliary graph,~$\Gamma^\mc{M} _A$, an alternating path is a sequence of edges starting from an unmatched node in~$\delta \mc{M}$ and every second edge in~$\mc{M}$, see red highlights in Fig.~\ref{fig3node} (d). The name comes from the alternating edges between unmatched part,~$\mc{E} \backslash \mc{M}$, and matched part,~$\mc{M}$, in~$\Gamma^\mc{M}_A$.

{\bf Contraction}: In the auxiliary graph,~$\Gamma^\mc{M} _A$, assign a contraction,~$\mc{C}_i$, to every unmatched node,~$v_i \in \delta \mc{M}$. The set,~$\mc{C}_i$, contains all states reachable in $\mc{V}^+$ by alternating paths starting from~$v_i$. Further, define~$\mc{C}$ as the set of all~$\mc{C}_i$'s. Intuitively, in a contraction, states are contracted to a fewer number of nodes, shown in Fig.~\ref{fig3node} (e) as a social digraph where now the contraction is highlighted in green. 
\begin{figure}
\centering
\subfigure{\includegraphics[height=1.7in]{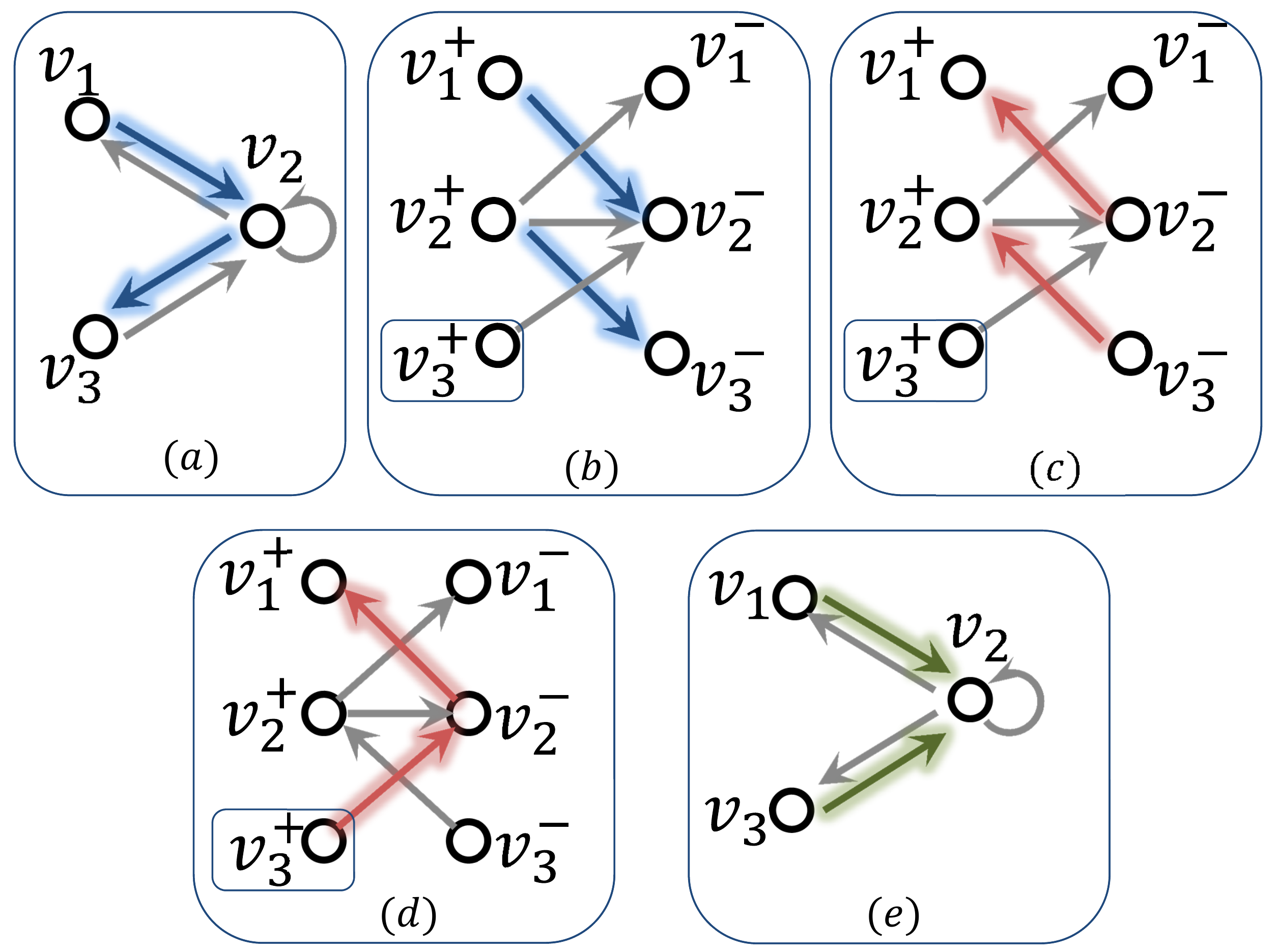}}
\hspace{0.3cm}
\subfigure{\includegraphics[height=1.7in] {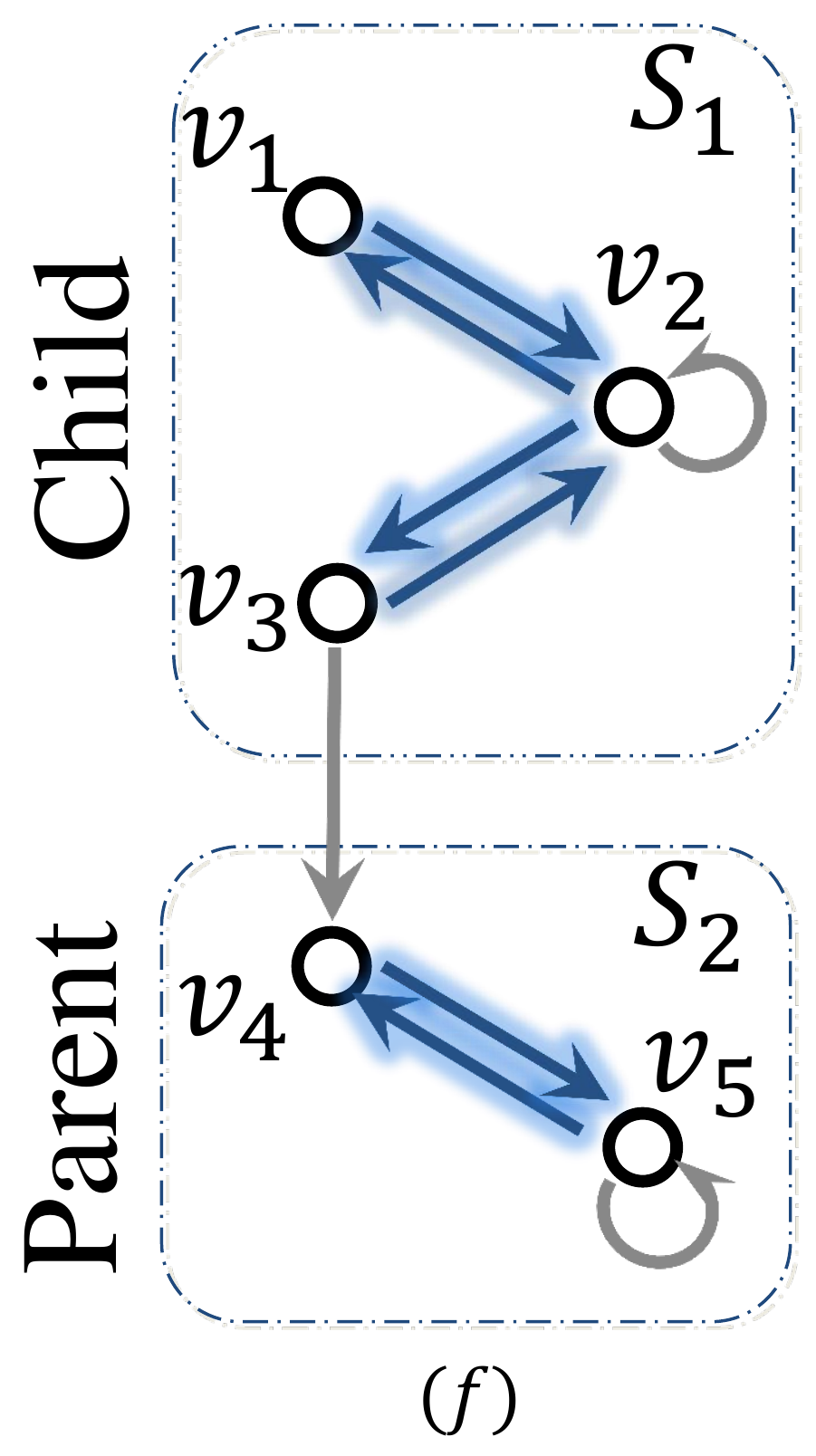}}
\caption{(a) A simple social digraph,~$\mc{G}_A$. (b) Bipartite graph,~$\Gamma_A$, obtained from~$\mc{G}_A$, where a maximal matching is highlighted in blue. (c) Auxiliary graph~$\Gamma^\mc{M}_A$; (d) Alternating path; (e) Contractions. (f) A social digraph to illustrate SCCs.}
\label{fig3node}
\end{figure}

\subsection{Strongly Connected Components}\label{agr2}
We illustrate on the social digraph,~$\mc{G}_A$ in Fig.~\ref{fig3node} (f).

{\bf Strong-connectivity}: A digraph is strongly-connected if every two nodes in the digraph are connected by a path, i.e.,~$v_i\overset{\scriptsize\mbox{path}}{\longrightarrow} v_j$ for every~$v_i, v_j$, in the digraph. 

{\bf SCC}: In a not strongly-connected digraph, a Strongly Connected Component (SCC),~$\mc{S}_i$, is defined as its maximal strongly-connected partitions. The highlighted blue edges in Fig.~\ref{fig3node} (f) represent strong connectivity in each SCC.

\textbf{Matched SCC}: An SCC,~$i$, is matched, denoted by~$\mc{S}^{\circlearrowleft}_i$, if it contains a union of \emph{disjoint} cycles covering all of its nodes. We denote the set of all matched SCCs by~$\mc{S}^{\circlearrowleft}$. In Fig.~\ref{fig3node} (f), $\mc{S}_1$ is un-matched and $\mc{S}_2$ is matched.

\textbf{Parent/Child SCC}: An SCC,~$i$, is \textit{parent}, denoted by~$\mc{S}^{p}_i$, if it has no outgoing edges to any other state in $\mc{G}_A$. Any non-parent SCC is \textit{child}, denoted by~$\mc{S}^{c}_i$. In Fig.~\ref{fig3node} (f),~$\mc{S}_1$ is child and~$\mc{S}_2$ is parent. Let~$\mc{S}^{p}$ be the set of all parent SCCs. Following this convention,~$\mc{S}_i^{\circlearrowleft,p}$ is the matched parent SCC,~$\mc{S}^{\circlearrowleft,p}$ is the set of matched parent SCCs, and so on.

\textbf{Partial order},~$\preceq$, defines the existence of edges within components. Mathematically,~$\mc{S}_i~\preceq \mc{S}_j$ implies that some node in~$\mc{S}_i$ has a path to some node in~$\mc{S}_j$. In Fig.~\ref{fig3node} (f),~$\mc{S}_1~\preceq \mc{S}_2$.

\begin{rem}\label{childSCC}
Every child SCC,~$\mc{S}^{c}_i$, has a parent SCC,~$\mc{S}^{p}_j$, i.e.,~$\mc{S}^c_i \preceq \mc{S} ^p_j$. Matched parent SCCs are defined over cyclic (matched) part of $\mc{G}_A$ and they lack any unmatched node. On the other hand, contractions are defined over unmatched part of $\mc{G}_A$, and include an unmatched node.
\end{rem}

\subsection{Computational Algorithms}\label{DMalg}
Given the system digraph,~$\mc{G}_A$ and its bipartite counterpart,~$\Gamma_A$, we can use efficient algorithms to compute the maximal matching,~$\mc{M}$, e.g. the maximum flow algorithm~\cite{hopcraft}. These three objects,~$\mc{G}_A,\Gamma_A,\mc{M}$, are then used in the Dulmage-Mendelsohn (DM) decomposition~\cite{dulmage58,murota} to obtain the set of contractions,~$\mc{C}$, and the set of matched parent SCCs,~$\mc{S}^{\circlearrowleft,p}$. Maximal matchings can be efficiently computed in~$\mc{O}(\sqrt{n} |\mc{E}_A|)$ using the approach in~\cite{maxmatching}. Efficient algorithms to decompose a digraph into maximal SCCs include the well-known \textit{Tarjan's algorithm}~\cite{tarjan} and related Depth-First-Search (DFS) algorithms~\cite{algorithm}. These algorithms have polynomial order in~$|\mc{E}_A|$.

\section{Necessary measurements for observability}  \label{nessmeas}
In this section, we find necessary observations for centralized observability, i.e. observability of~$(A,H)$. More precisely, we look for the states in the social digraph,~$\mc{G}_A$, whose observations ensure generic observability. We approach this problem in two stages to recover both conditions in Theorem~\ref{1}. First, we find necessary agents to meet the accessibility property and, second, we look for agents recovering the~$S$-rank.

\subsection{Recovering accessibility}
It is known that the accessibility of system states has a direct connection with SCCs in the social digraph~\cite{asilomar11}. The reason is that states in SCC are all accessible to each-other; if an agent,~$i$, is accessible to a state in SCC,~$\mc{S} _{j}$, it is accessible to
all other states in~$\mc{S}_{j}$. Every state~$x_k$ included in~$\mc{S} _{j}$ is connected to agent,~$i$, via a path of state nodes, i.e.~$x_k \overset{\scriptsize\mbox{path}}{\longrightarrow} \mc{Y}$.

\begin{theorem} \label{parent}
At least one observation from every matched parent SCC is necessary to recover the observability of the social digraph,~$\mc{G}_A$, .
\end{theorem}
\begin{proof}
If parent SCC,~$\mc{S}^{\circlearrowleft,p}_i$, has no outgoing edges, adding an observation is the only way to recover the accessibility of its states, see~\cite{asilomar11} for more details.
\end{proof}

\begin{cor}
If a social digraph is observable, then every matched parent SCC,~$\mc{S}^{\circlearrowleft,p}_j$, has a link to an agent in~$\mc{G}_{\scriptsize \mbox{sys}}$. Hence, all~$\mc{S}^{\circlearrowleft,c}_i$'s are accessible:~$\mc{S}^{\circlearrowleft ,c}_i \overset{\scriptsize\mbox{path}}{\longrightarrow} \mc{S}^{\circlearrowleft,p}_j \overset{\scriptsize\mbox{path}}{\longrightarrow} \mc{Y}$.
\end{cor}

\begin{defn}
We name the observations in matched parent SCCs as \textit{Type-$\beta$}, and agents measuring them as~$\beta$-agents.
\label{beta}
\end{defn}
\begin{rem}
To recover accessibility, agents observing different states of the same parent SCC belong to an \emph{equivalence class}\footnote{An equivalence relation is defined as having three requirements of reflexivity, symmetry, and transitivity.} and are called equivalent. They are called independent if their observations are from distinct parent SCCs.
\end{rem}

\subsection{Recovering~$S$-rank}
We now characterize the necessary observations for recovering the~$S$-rank condition. Specifically, we show that adding observation in contraction sets recovers the~$S$-rank of~$\left[ \begin{array}{cc} A^\top & H^\top \end{array} \right]^\top$. To study the contraction property of the system we review some useful lemmas below.

\begin{lem}\label{linL}
The~$S$-rank of~$\left[ \begin{array}{cc} A^\top & H^\top \end{array} \right]^\top$ equals the size of a maximal matching~$\mc{M}$ in its composite digraph,~$\mc{G}_{sys}$, see~\cite{Liu-nature}.
\end{lem}

\begin{lem}
Any choice of maximal matching gives the same contraction set. Mathematically, having two maximal matchings,~$\mc{M}_1 \neq \mc{M}_2$, any unmatched state~$v_i \in \delta \mc{M}_1$ can be reached along an alternating path from a state~$v_j \in \delta \mc{M}_2$; the set~$\mc{C}$ is the same for both of matchings, see~\cite{murota, Berge}.
\end{lem}

\begin{lem} \label{berge}
For any choice of maximal matching,~$\mc{M}$, there is only one unmatched node in every contraction,~$\mc{C}_i$, see~\cite{Berge}.
\end{lem}

\begin{lem} \label{sizerecovery}
Adding an observation of an unmatched node in~$\delta \mc{M}$ recovers the $S$-rank by~$1$,~\cite{murota}.
\end{lem}

\textit{Example:} We illustrate Lemmas~\ref{linL}--\ref{sizerecovery} in Fig.~\ref{figcontraction} where a contraction of~$3$ nodes,~$x_1,x_3,x_5$, into~$2$ nodes,~$x_2,x_4$, is shown. The number of possible maximal matchings is~${3\choose2}=3$. From Fig.~\ref{figcontraction}, a maximal matching gives one unmatched node in the contraction: e.g. in Fig.~\ref{figcontraction}~(b),~$x_1$ is the unmatched node, and after reversing the (highlighted) edges in that maximal matching, nodes~$x_3$ and~$x_5$ are reachable from~$x_1$. Simarly, Figs.~\ref{figcontraction} (c) and (d) show the remaining maximal matchings. 
\begin{figure}
\centering
\includegraphics [width=2.5in]{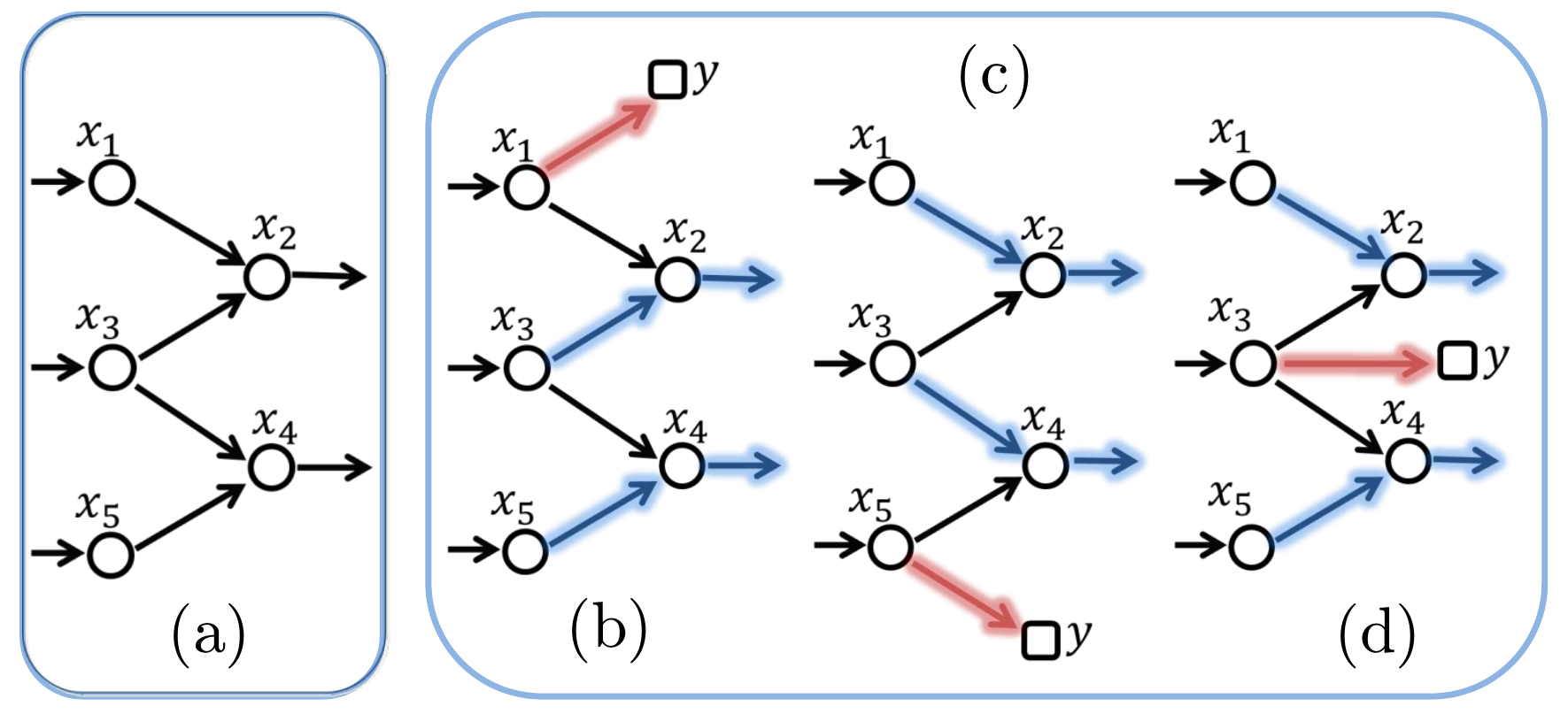}
\caption{Possible maximal matchings (shaded arrows) in a contraction where observation are added tp the unmatched node shown as red shaded arrow.}
\label{figcontraction}
\end{figure}

\begin{prop} \label{rank-recovery}
Let a matrix~$H_{\mc{C}_i}$ define an observation from an unmatched node in a contraction~$\mc{C}_i$, then
\begin{eqnarray}
S\mbox{-rank}\left[
\begin{array}{c} A \\ H_{\mc{C}_i}
\end{array} \right] = S\mbox{-rank}~(A) +1.
\end{eqnarray}
\end{prop}
\begin{proof}
Directly follows from Lemma~\ref{berge} and~\ref{sizerecovery}. Since there is an unmatched node, note that $S\mbox{-rank}(A)<n$. For any choice of maximal matching~$\mc{M}$ (say highlighted edges in Fig.~\ref{figcontraction} (b)) there is only one unmatched node, ($x_1\in\mc{C}_i=\{x_1,x_3,x_5\}$), where other states in~$\mc{C}_i \backslash x_1$ are all matched. Therefore, observing~$x_1$ improves the~$S$-rank by~$1$.
\end{proof}
\begin{theorem} \label{contraction}
To recover the~$S$-rank of~$A$, one distinct state observation from every contraction set~$\mc{C}_i$ is necessary.
\end{theorem}
The proof follows from the previous arguments. For~$S$-rank recovery, the observations from any state in a contraction,~$\mc{C}_i$, are \textit{equivalent}. In other words, for any state in the contraction, there exists a maximal matching where this state is unmatched. If two contraction sets, say~$\mc{C}_i$ and~$\mc{C}_j$, share a state, observation from that common state only recovers the~$S$-rank of one of them. Hence, a \textit{distinct} observation from the other contraction set is required. The observations are called \textit{independent} if they are distinct and are taken from distinct contractions.

\begin{defn} \label{alpha}
We name the observations in the contractions as \textit{Type-$\alpha$}, and the agents measuring them as~$\alpha$-agents.
\end{defn}

Notice that the necessary observations for~$S$-rank recovery are different from Type-$\beta$ observations recovering accessibility. Type-$\alpha$ and Type-$\beta$ observations are both equally critical for centralized observability; however, they play different roles in distributed case as will be discussed in Section~\ref{partial}. To the best of our knowledge, this is not considered in the literature and makes this work of more interest as compared to centralized observability recovering in~\cite{Liu-nature,commault-recovery ,boukhobza-recovery}.

\subsection{Necessary conditions for centralized observability} \label{nesscon}
We now merge the necessary conditions from the previous in the following theorem as the main result of this section.
\begin{theorem} \label{centralized}
For a social digrpah,~$\mc{G}_A$, the observations necessary for generic observability are:
\begin{enumerate} [(i)]
\item one from every contraction set,~$\mc{C}_i$;
\item one from every matched parent SCC,~$\mc{S}^{\circlearrowleft,p} _{i}$.
\end{enumerate}
\end{theorem}
\begin{proof}
The proof follows from Theorems~\ref{parent} and~\ref{contraction}.
\end{proof}

Theorem~\ref{centralized} is not only true for the centralized case where all observations are collected at a central coordinator as in~\cite{commault-recovery, boukhobza-recovery}, but is also applicable to the distributed case. If each agent has access to these necessary observations it can infer the global information of all social states. However, in the distributed case, agents have partial measurements and thus, the idea is to recover these necessary conditions by communicating and sharing information among the agents.

\section{Recovering local observability at each agent} \label{partial}
In this section, we extend the results to distributed estimator, where the observations are available over a network. Through the links in the network, agents can share necessary information on their observations and/or predictions, and recover their partial observability. Assuming no information loss over the communication links, the chief objective is to design the necessary topology of the agent network to ensure distributed (generic) observability. Recall that by distributed observability we imply the observability of the pair~$(W \otimes A , D_H)$.

\begin{theorem}\label{mainthm}
Consider the composite digraph,~$\mc{G}_{\mbox{\scriptsize sys}}$, to have the necessary observations,~$H$, from Theorem~\ref{centralized}. The digraph possesses generic  distributed observability if and only if every agent,~$i$, has the following characteristics:
\begin{enumerate}[(i)]
\item For every contraction,~$\mc{C}_l$, agent~$i$ receives a direct link from an~$\alpha$-agent measuring a state in~$\mc{C}_l$;
\item Either one of the following for every~$\mc{S}^{\circlearrowleft , p}_l$:
\begin{enumerate}[(a)]
\item Agent~$i$ receives a direct link from a~$\beta$-agent measuring a state in~$\mc{S}^{\circlearrowleft , p}_l$;
\item Agent~$i$ sends its information via a sequence of agents making a directed path to a~$\beta$-agent, say $j$, with a state observation in~$\mc{S}^{\circlearrowleft,p}_l$.
\end{enumerate}
\end{enumerate}
\end{theorem}
\begin{proof}
Sufficiency is proved in~\cite{jstsp}. Necessity follows a similar argument: The proof of parts (i) and (ii)-(a) comes directly from Theorem~\ref{centralized}; in part (i), receiving a state observation from every contraction~$\mc{C}_l$ recovers the~$S$-rank condition at each agent, while, in part (ii)-(a), receiving a state observation from every~$\mc{S}^{\circlearrowleft,p}_l$ directly recovers the accessibility at agent~$i$. Part (ii)-(b) recovers the accessibility in~$W \otimes A$ indirectly. A directed path from agent~$i$ to~$\beta$-agent~$j$ makes the inaccessible SCC~$\mc{S}^{\circlearrowleft ,p}_l$ at agent~$i$, accessible via agent~$j$.
\end{proof}

Condition (i) defines an~$\alpha$-network,~$\mc{G}_\alpha$, where agents share measurements directly with each other. Whereas, condition (ii)-b defines a~$\beta$-network,~$\mc{G}_\beta$, over which the agents {\bf\textit {only}} share their predictions. Notice that, this connectivity requirement is weaker than the necessary condition in~\cite{nuno-suff.ness}, where each agent requires to transmit/share {\bf \textit{both}} its observations and predictions to every other agent over the same network. It is noteworthy that (ii)-(a) is a straightforward technique to recover accessibility~\cite{sauter:09}. But it may require long-distance links as compared to (ii)-(b). In particular, when the system is full~$S$-rank there are no Type-$\alpha$ agents and any strongly-connected network is sufficient (not necessary) to satisfy (ii)-(b). This assumption is prevalent to guarantee stability of distributed estimators, e.g. in~\cite{usman_cdc:11, sayedtu12, battistelli_cdc}, however, as we have shown, it is only applicable to full~$S$-rank systems.

{\bf Minimal observability}: We now consider the \textit{minimal} conditions for distributed observability, i.e. minimal number of agents and their connectivity. When the matched parent SCCs and contractions share some states, observation of those common states recovers both~$S$-rank and accessibility. Thus, it reduces the minimal number of necessary agents. Here, we first define the necessary number of different type of agents. From Proposition~\ref{rank-recovery}, each Type-$\alpha$ observation recovers one~$S$-rank deficiency of the social digraph,~$A$. Similarly, since all parent SCCs are disjoint, every Type-$\beta$ observation recovers the accessibility of one matched parent SCC.
\begin{rem}
Let~$n_\alpha$ and~$n_\beta$ denote the number of necessary Type-$\alpha$ and Type-$\beta$ observations, respectively. Then,
\begin{equation}
n_\alpha = |\mc{C}| = n - S\mbox{-rank} (W\otimes A), \qquad
n_\beta = \left|\mc{S}^{\circlearrowleft,p}\right|.
\end{equation}
\end{rem}
Note that $S\mbox{-rank}(W\otimes A)=S\mbox{-rank}(A)$, because $W$ is full $S\mbox{-rank}$~\cite{asilomar11,icassp13}. Further, the number,~$n_\beta$, of necessary Type-$\beta$ observations is to be corrected as the sets~$\mc{S}^{\circlearrowleft,p}$ and~$\mc{C}$ may share some states--e.g. see Section~\ref{example}. Adjusting for the possible shared states, we have the following lemma.
\begin{theorem}
The minimum number of necessary observations for distributed observability is~$|\mc{C}| +|\mc{S}^{\circlearrowleft,p}|-|\mc{S}^{\circlearrowleft,p} \cap \mc{C}|.$
\end{theorem}
The proof is straightforward and follows the previous discussion. Notice that~$\alpha$-agents necessitate \textit{strict} connectivity in the network as compared to~$\beta$-agents. Recovering the~$S$-rank is \textit{only} possible via a direct observation; however, the accessibility can be obtained via a direct observation \textit{or} a directed path to a~$\beta$-agent, where the latter case gives minimal connectivity requirements for~$\beta$-agents.

\begin{rem}
Condition (ii)-(b) in Theorem~\ref{mainthm} gives the \textit{minimal} connectivity requirement for~$\beta$-agents. The~$\alpha$-agents are more critical than the~$\beta$-agents because of their stringent connectivity. This implies that observation of any state in~$\mc{S}^{\circlearrowleft,p} \cap \mc{C} \neq \emptyset$ are of Type-$\alpha$. In other words,~$n_\beta = |\mc{S}^{\circlearrowleft,p}| - |\mc{C} \cap \mc{S}^{\circlearrowleft,p}|$. Further note that the observability of $(W \otimes A,D_H)$ is tied to the structure, $A$, of the social digraph,~$\mc{G}_A$,~\cite{icassp13}. Particularly, the S-rank of~$A$ defines the number of Type-$\alpha$ observations as discussed above. This, in turn, affects the agent network,~$\mc{G}_W$, i.e. the structure of the matrix~$W$, as provided in Theorem~\ref{mainthm}.
\end{rem}

\section{Illustrations} \label{example}
We provide a simple example to explain the concepts and results established in this paper. Consider an~$n=6$-state system, whose digraph,~$\mc{G}_A$, associated bipartite graph,~$\Gamma_A$, and auxiliary graph,~$\Gamma^\mc{M}_A$, are shown in Fig.~\ref{figgraph}.
\begin{figure}
\centering
{\includegraphics[width=3.2in]{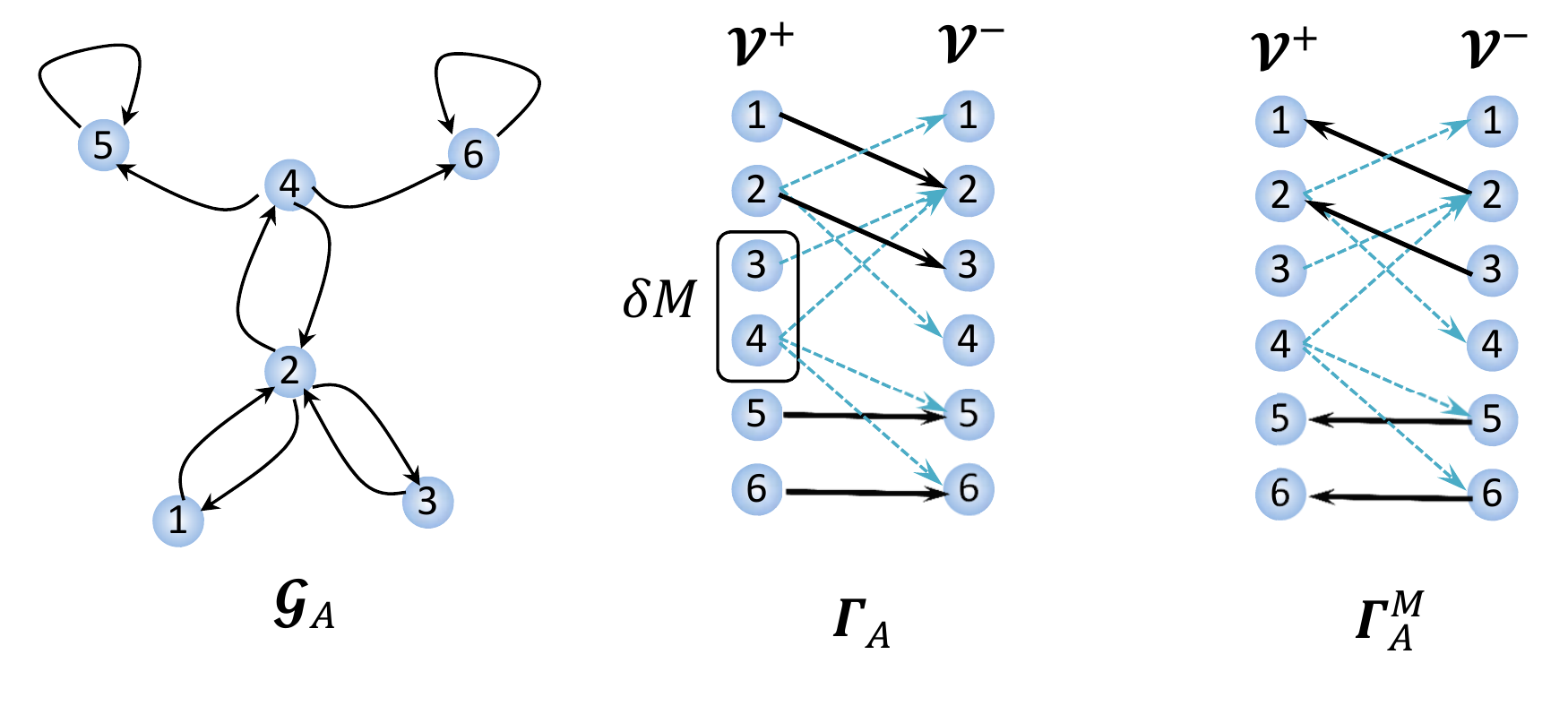}}
\caption{Social digraph,~$\mc{G}_A$, Bipartite graph,~$\Gamma_A$, and Auxiliary graph,~$\Gamma_A^{\mc{M}}$}
\label{figgraph}
\end{figure}
The maximal matching,~$\mc{M}$, is shown as black edges in~$\Gamma_A$ and~$\Gamma^\mc{M}_A$. The contractions and unmatched nodes are illustrated in Fig.~\ref{figG} (left).
\begin{figure}
\centering
{\includegraphics[width=3.5in]{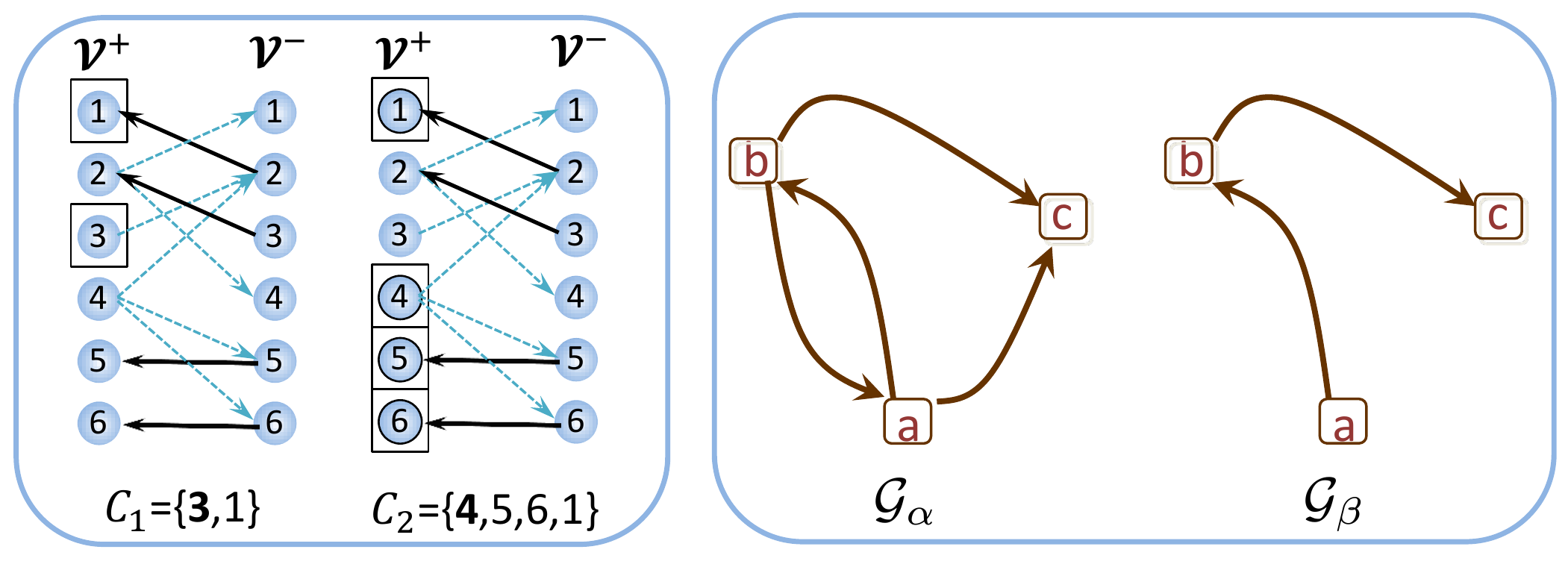}}
\caption{(Left) Contraction sets of the system in Fig.~\ref{figgraph}. (Right) Necessary communication network:~$\mc{G}_W=\mc{G}_\alpha \cup \mc{G}_\beta$.}
\label{figG}
\end{figure}
The unmatched nodes are~$\delta \mc{M} = \{3 , 4 \}$. The contractions are determined via the alternating paths in auxiliary graph as~$\mc{C} = \{\{3,1\},\{4,5,6,1\}\}$. Both parent SCCs are  matched (cyclic), thus~$\mc{S}^{\circlearrowleft,p}=\{\{5\},\{6\}\}$.

It can be verified that the~$S$-rank of the associated system matrix is~$4$. From Theorem~\ref{centralized}, we require~$n_\alpha=6-4=2$ observations from~$\mc{C}$, and~$n_\beta=2$ observations from~$\mc{S}^{\circlearrowleft,p}$. Note that~$\mc{C} \cap \mc{S}^{\circlearrowleft,p} = \{5,6\}$, therefore, having a Type-$\alpha$ observation of~$x_5$ or~$x_6$ recovers accessibility of~$\mc{S}^{\circlearrowleft,p}_1 = \{5\}$ or~$\mc{S}^{\circlearrowleft,p}_2 = \{6\}$, respectively. We get~$\min\{n_\beta\} = 2-1= 1$. We choose the following observations: Type-$\alpha$ agents,~$a,b$, observing~$x_3, x_5$; Type-$\beta$ agent~$c$ observing~$x_6$. Next, we define the necessary network connectivity from Theorem~\ref{mainthm}: (i) agents~$a$ and~$b$ send their observations directly to each other, and agent~$c$, over~$\mc{G}_\alpha$; and, (ii) agent~$c$ receives information from agents~$a$ and~$b$ over~$\mc{G}_\beta$.

\subsection{Large-scale Social Networks}\label{lssn}
We now provide some insights of our results towards inference in social networks. Consider a social group of actors with states, e.g. opinions, sentiments, emotions, etc., that evolve over social interactions. The influence network, e.g. friendship, co-authorship, swarming, etc., is time-invariant but the influence weight of actors may vary over time, and different weight assignment model the evolution of different states resulting into different social phenomena. Our aim is to infer such phenomena by observing some necessary states without considering any particular dynamics but only the social interactions (digraph). For distributed inference, first, we classify these states (and the agents observing them) according to Definitions~\ref{beta} and \ref{alpha}. The necessary network of agents is defined according to Theorem~\ref{mainthm}. 

Following the discussion in Section~\ref{mdl}, the structure of any social digraph is highly relevant to the dynamics that may take place over the social network. In this context, we use some of the well-known social network models~\cite{DuBois:2008,monksdata} and explore the graphical observability results developed in this paper. These networks have been used for the estimation of corresponding social phenomena modeled as LSI systems as discussed in Section~\ref{mdl}. Each node (circles in Figs.~\ref{sampson}--\ref{sn_fig4}) represents a state, e.g., heading, opinion, buying habits, etc., in the social digraph and evolves over social interactions. Theorem~\ref{centralized} characterizes the necessary observerations. These observations (and their associated agents) are classified as Type-$\alpha$ (red circles) and Type-$\beta$ (green circles). Finally, Theorem~\ref{mainthm} characterizes the network of these agents accordingly. The results are summarized in the table below:
\begin{center}
\begin{tabular}
{ | c | c | c | c | c |} \hline
 & $n=|\mc{V}_A|$ & $E=|\mc{E}_A|$  & $n_\alpha$ \tikz\draw[red,fill=red] (0,0) circle (.5ex); & $n_\beta$ \tikz\draw[green,fill=green] (0,0) circle (.5ex);\\ \hline
Monks & $18$ & $88$  & $0$ & $1$ \\ \hline
Blogs & $1224$ & $19025$  & $436$ & $0$ \\ \hline
Books & $105$ & $882$  & $0$ & $1$ \\ \hline
Coauthorship & $1461$ & $5484$  & $37$ & $248$ \\ \hline
\hline
\end{tabular}
\end{center}

(a) \emph{Sampson's Monastery Network}, mentioned earlier in Section~\ref{mdl}, is a directed network of interactions among the Monks in a monastery. The digraph from~\cite{monksdata} is shown in Fig.~\ref{sampson}. The network is full $S$-rank, implying $n_\alpha=0$, and is strongly-connected so $n_\beta=1$. To illustrate agent connectivity, assume a collection of such monasteries, each with one necessary Type-$\beta$ measurement monitored by one $\beta$-agent. From our results, it is necessary for the agents to communicate over a strongly-connected network in order to estimate any social phenomena on the union of the corresponding social digraphs.
\begin{figure}[!h]
\centering
\includegraphics[width=2in]{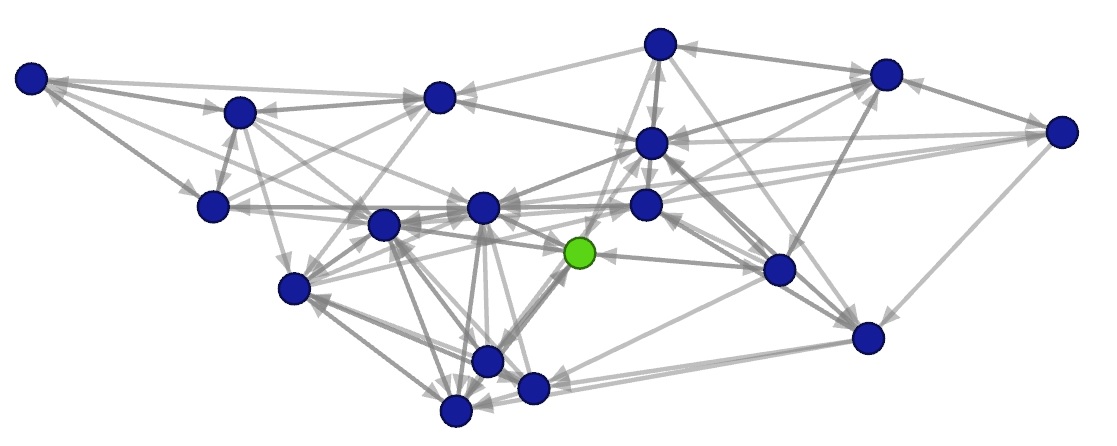}
\caption{Directed Sampson's network with $18$ actors.}
\label{sampson}
\end{figure}

(b) \emph{Political Blogs}: A social digraph of hyperlinks between weblogs on US politics~\cite{uci_polblogs}, shown in Fig.~\ref{sn_fig23} (left). Each node represents a blog linked to other political blogs; the state at each node could be the popularity of the blog evolving via political commentary~\cite{Diggblog}. The blogs can be seen to have two dominant clusters constituting blogs that are more followed and hyperlinked. The digraph has $n_\alpha=436$ unmatched nodes. We may observe that most of the Type-$\alpha$ agents appear on the boundary of the network where the blogs are less cited (hyperlinked), and thus, may not be inferred from the interior nodes. This specific example of inference of the popularity of such blogging network shows that: (i) hubs (nodes with high degrees) are not necessary for observability; and, (ii) to extract the popularity of all blogs in a distributed way, a fully-connected network is necessary (and sufficient \cite{jstsp}).

(c) \emph{Books on US Politics}: Amazon.com data--undirected edges represent co-purchasing of books by the same buyers~\cite{polbooks}, digraph is shown in Fig.~\ref{sn_fig23} (right). The network has full $S$-rank, thus $n_\alpha=0$, and is further connected so $n_\beta=1$, and can be an observation from any node.
\begin{figure}[!h]
\centering
\subfigure{\includegraphics[width=1.5in]{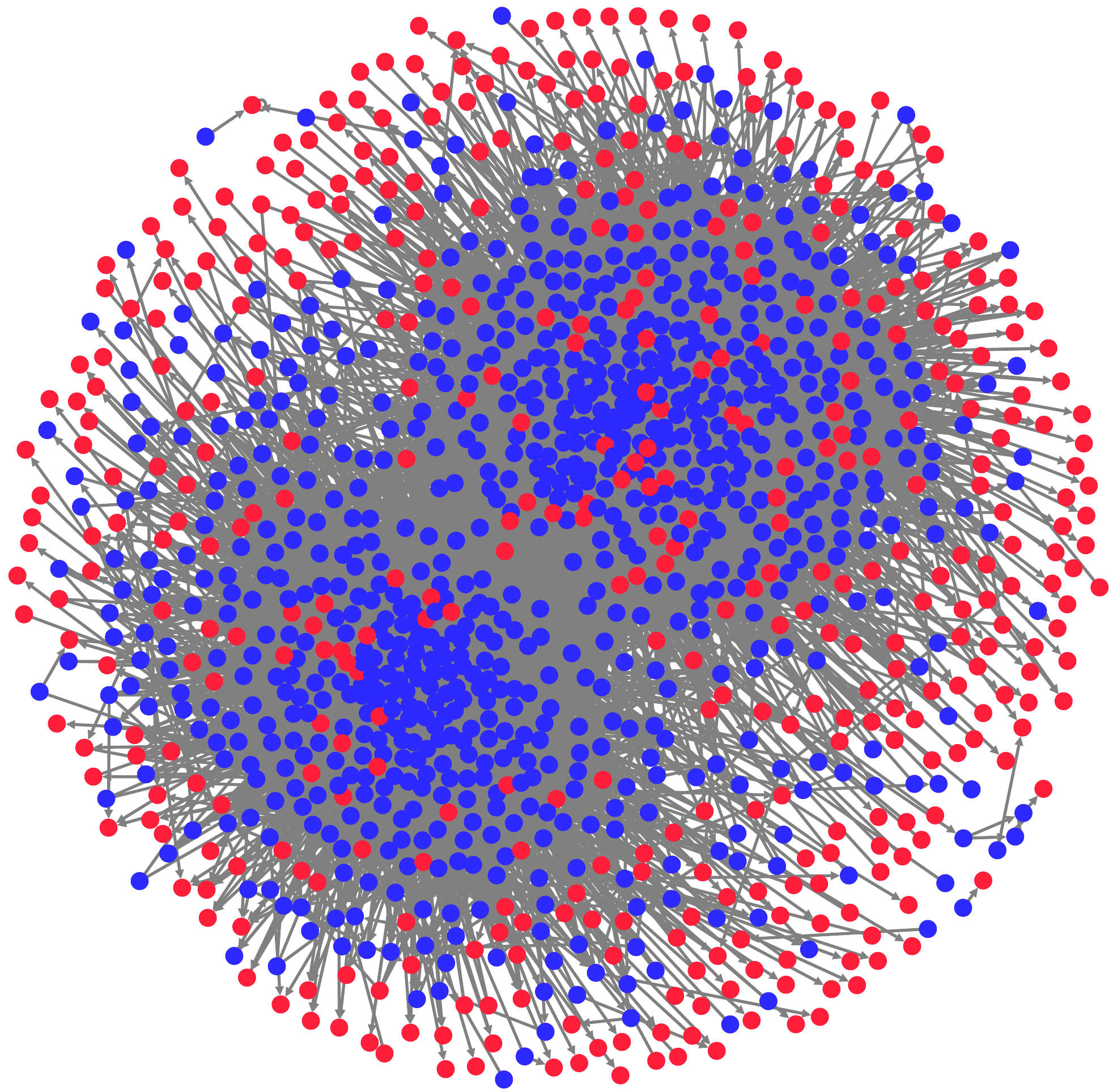}}
\hspace{1cm}
\subfigure{\includegraphics[width=1.5in]{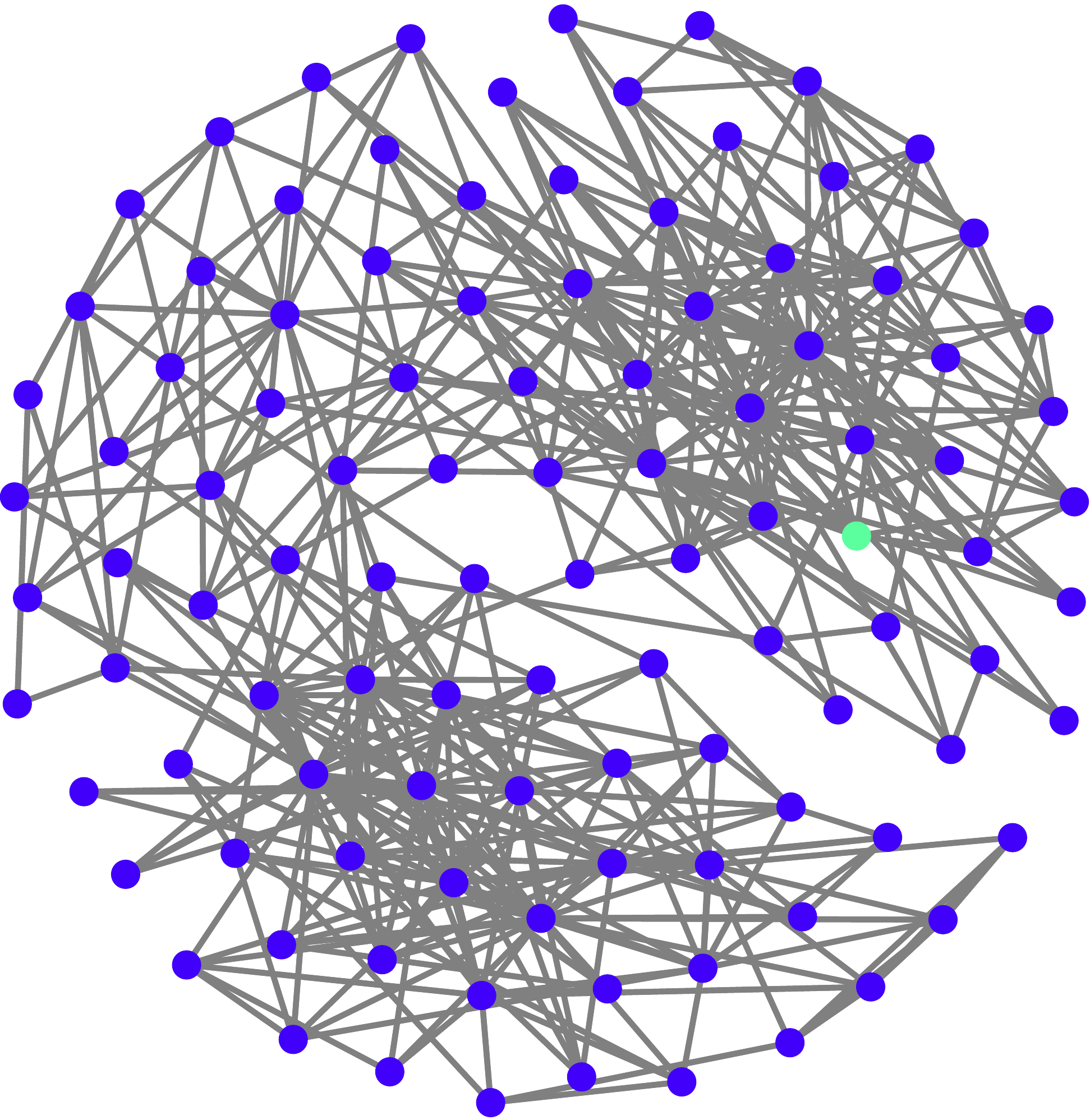}}
\caption{Social digraphs: (Left) Political blogs during the $2004$ US Elections with $1224$ nodes; (Right) Network of political books with $105$ nodes.}
\label{sn_fig23}
\end{figure}

(d) \emph{Co-authorship in Network Science}: A digraph of researchers in network theory~\cite{uci_netsci}, shown in Fig.~\ref{sn_fig4}. The states  may model a novel concept or a result and the links represent the influence among the authors. The digraph contains $268$ components out of which $248$ are matched. All of matched components are parent resulting into $n_\beta=248$; and, $n_\alpha=37$. Wiring according to Theorem~\ref{mainthm}, each agent may infer any phenomena that evolves over this social digraph.
\begin{figure}[!h]
\centering
\includegraphics[width=3in]{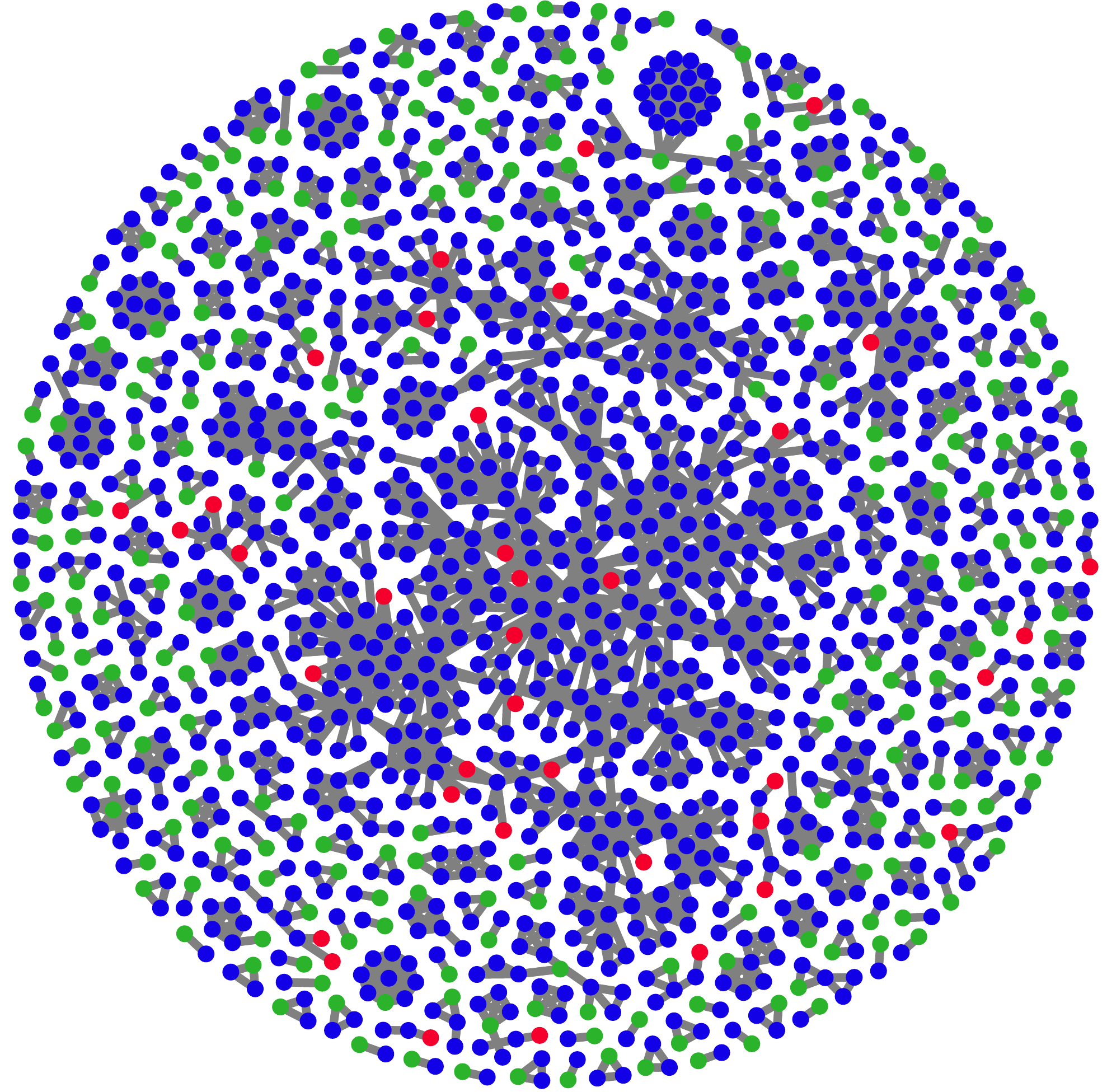}
\caption{Co-authorship network with $1461$ nodes}
\label{sn_fig4}
\end{figure}

\section{Conclusions}\label{con}
This paper formally develops the necessary conditions for distributed observability of social networks modeled as LSI systems. We characterize necessary observations, agent classification, and network connectivity that enable each agent to infer any social phenomena evolving over a given social digraph. In particular, we show that the distributed observability requires no more observations than the centralized case; however, it necessitates certain classification and connectivity requirements on the agents observing those states. We partition the necessary agents to Type-$\alpha$ with strict connectivity requirements, and Type-$\beta$ with milder connectivity. We provide combinatorial algorithms to define such partitioning and show the relevance and applicability to real world examples of large-scale social systems. These results can be applied in distributed estimation of smart grids and other physical systems.

\bibliographystyle{IEEEbib}
\bibliography{bibliography}
\end{document}